\documentclass[10pt,reqno]{amsart}
\usepackage{latexsym,amsmath}
\usepackage{times}
\usepackage{amsfonts}
\usepackage{amssymb}
\usepackage{latexsym}

\usepackage{bbm}

\usepackage{color}
\usepackage{fullpage}

\newcommand\nix{\,\cdot\,}
\newcommand\bis{\mathrm{bis}}
\newcommand\cut{\mathrm{cut}}

\newcommand\ism{\cong}

\newcommand\G{\vec G}
\newcommand\T{\vec T}

\numberwithin{equation}{section}

\newcommand\br[1]{\left(#1\right)}

\renewcommand{\vec}[1]{\boldsymbol{#1}}

\DeclareMathOperator{\pr}{\mathbb P}

\newcommand\SIGMA{\vec\sigma}

\newcommand\TAU{\vec\tau}

\newcommand{\toboss}{\uparrow}

\newtheorem{definition}{Definition}[section]

\newtheorem{theorem}[definition]{Theorem}
\newtheorem{lemma}[definition]{Lemma}
\newtheorem{proposition}[definition]{Proposition}
\newtheorem{corollary}[definition]{Corollary}

\newcommand\core{\cC}

\newcommand\cC{\mathcal{C}}

\newcommand\cU{\mathcal{U}}

\newcommand\cS{\mathcal{S}}
\newcommand\cT{\mathcal{T}}

\newcommand\cL{\mathcal{L}}

\newcommand\cP{\mathcal{P}}
\newcommand\cX{\mathcal{X}}

\def\cC{{\mathcal C}}

\newcommand\eps{\varepsilon}

\newcommand\Erw{\mathbb{E}}

\newcommand{\vecone}{\vec{1}}

\newcommand{\Po}{{\rm Po}}
\newcommand{\Bin}{{\rm Bin}}

\newcommand\bc[1]{\left({#1}\right)}
\newcommand\cbc[1]{\left\{{#1}\right\}}

\newcommand\brk[1]{\left\lbrack{#1}\right\rbrack}

\newcommand\abs[1]{\left|{#1}\right|}

\newcommand\RR{\mathbb{R}}

\newcommand{\whp}{w.h.p.}

\newcommand{\stacksign}[2]{{\stackrel{\mbox{\scriptsize #1}}{#2}}}

\newcommand{\Erdos}{Erd\H{o}s}
\newcommand{\Renyi}{R\'enyi}

\newcommand\Lem{Lemma}
\newcommand\Prop{Proposition}
\newcommand\Thm{Theorem}

\newcommand\Sec{Section}

\newcommand{\pplus}{p_{+}}
\newcommand{\pminus}{p_{-}}
\newcommand{\dplus}{d_{+}}
\newcommand{\dminus}{d_{-}}
\newcommand{\ppm}{p_{\pm }}
\newcommand{\dpm}{d_{\pm }}
\newcommand{\gplus}{\gamma_{+}}
\newcommand{\gminus}{\gamma_{-}}
\newcommand{\gpm}{\gamma_{\pm }}
\newcommand{\partialpm}{\partial_{\pm }}

\newcommand{\splus}{S_{+}}
\newcommand{\sminus}{S_{-}}
\newcommand{\cpm}{C_{\pm }}
\newcommand{\cplus}{C_{+}}
\newcommand{\cminus}{C_{-}}
\newcommand{\zpplus}{Z_{p,+}}
\newcommand{\zpminus}{Z_{p,-}}

\newcommand{\cpl}{c^*}

\begin{document}

\title{The Minimum Bisection in the Planted Bisection Model}

\author{Amin Coja-Oghlan$^*$, Oliver Cooley$^{**}$, Mihyun Kang$^{**}$ and Kathrin Skubch}
\thanks{$^*$The research leading to these results has received funding from the European Research Council under the European Union's Seventh 
Framework Programme (FP/2007-2013) / ERC Grant Agreement n.\ 278857--PTCC\\
$^{**}$Supported by Austrian Science Fund (FWF): P26826 and W1230, Doctoral Program ``Discrete Mathematics''.}
\date{\today}

\address{Amin Coja-Oghlan, {\tt acoghlan@math.uni-frankfurt.de}, Goethe University, Mathematics Institute, 10 Robert Mayer St, Frankfurt 60325, Germany.}

\address{Oliver Cooley, {\tt cooley@math.tugraz.at}, Graz University of Technology, Institute of Optimization and Discrete Mathematics (Math B), Steyrergasse 30, 8010 Graz, Austria}

\address{Mihyun Kang, {\tt kang@math.tugraz.at}, Graz University of Technology, Institute of Optimization and Discrete Mathematics (Math B), Steyrergasse 30, 8010 Graz, Austria}

\address{Kathrin Skubch, {\tt skubch@math.uni-frankfurt.de}, Goethe University, Mathematics Institute, 10 Robert Mayer St, Frankfurt 60325, Germany.}

\maketitle

\begin{abstract}
\noindent
In the {\em planted bisection} model a random graph $\G(n,\pplus,\pminus )$ with $n$ vertices
is created by partitioning the vertices randomly into two classes of equal size (up to $\pm1$).
Any two vertices that belong to the same class are linked by an edge with probability $\pplus$ and any two
that belong to different classes with probability $\pminus <\pplus$ independently.
The planted bisection model has been used extensively to benchmark graph partitioning algorithms.
If $\ppm =2\dpm /n$ for numbers $0\leq \dminus <\dplus $ that remain fixed as $n\to\infty$, then
\whp\ the ``planted'' bisection (the one used to construct the graph) will not be a minimum bisection.
In this paper we derive an asymptotic formula for the minimum bisection width
 under the assumption that $\dplus -\dminus >c\sqrt{\dplus \ln \dplus }$
for a certain constant $c>0$.

\end{abstract}

\section{Introduction}\label{Sec_intro}

\subsection{Background and motivation}
Since the early days of computational complexity graph partitioning problems have played a central role in computer science~\cite{GJS,Karp}.
Over the years they have inspired some of the most important algorithmic techniques that we have at our disposal today,
	such as network flows or semidefinite programming~\cite{ARV,FKBisect,GW,Karpinski,Harry}.

In the context of the probabilistic analysis of algorithms, it is hard to think of a more intensely studied problem than
the {\em planted bisection model}.
In this model a random graph $\G=\G(n,p_{+1},p_{-1})$ on  $[n]=\{1,\ldots,n\}$ is created by choosing a map $\SIGMA:V\to\{-1,1\}$ uniformly at random subject to
$||\SIGMA^{-1}(1)|-|\SIGMA^{-1}(-1)||\leq1$ and connecting 
any two vertices $v\neq w$ with probability $p_{\SIGMA(v)\SIGMA(w)}$ independently,
where $0\leq p_{-1}<p_{+1}\leq 1$. To ease notation, we often write $\pplus$ for $p_{+1}$ and $\pminus$ for $p_{-1}$, and handle subscripts similarly for other parameters.

Given the random graph $\G$ (but not the planted bisection $\SIGMA$), the task is to find a {\em minimum bisection} of  $\G$, i.e., to partition
the vertices into two disjoint sets $S,\bar S=[n]\setminus S$ whose sizes satisfy $||S|-|\bar S||\leq1$ such that the number of $S$-$\bar S$-edges is minimum.
The planted bisection model has been employed to gauge algorithms based on spectral, semidefinite programming, flow and local search techniques, to name
but a few~\cite{BBScott,Boppana,Bui,Carson,Bisect,CondonKarp,Dimitriou,DyerFrieze,Feige,JerrumSorkin,Juels,Kucera,McSherry,Makarychev}.

Remarkably, for a long time the algorithm with the widest
range of $n,\ppm$ for which a minimum bisection can be found efficiently was one of the earliest ones, namely Boppana's spectral algorithm~\cite{Boppana}.
It succeeds if $$n(\pplus -\pminus )\geq c\sqrt{n\pplus \ln n}$$ for a certain constant $c>0$.
Under this assumption the planted bisection is minimum \whp\
In fact, recently the critical value $\cpl>0$ for which this statement is true was identified explicitly~\cite{MNS15}.
In particular, for $n(\pplus -\pminus )>\cpl \sqrt{n\pplus \ln n}$ \ the minimum bisection width simply equals $(\frac14+o(1))n^2\pminus $ \whp\

But if $n(\pplus -\pminus )< \cpl \sqrt{n\pplus \ln n}$, then the minimum bisection width will be strictly smaller than the width of the planted bisection \whp\
Yet there is another spectral algorithm~\cite{Bisect} that finds a minimum bisection \whp\ under the weaker assumption that
	\begin{equation}
	n(\pplus -\pminus )\geq c\sqrt{n\pplus \ln(n\pplus )},
	\end{equation}
for a certain constant $c>0$, and even certifies the optimality of its solution.
However, \cite{Bisect} does not answer what is arguably the most immediate question:
what is the {\em typical} value of the minimum bisection  width?

In this paper we derive the value to which the (suitably scaled) minimum bisection width converges in probability.
We confine ourselves to the case that $\frac n2 \ppm=\dpm$ remain fixed as $n\to\infty$.
Hence, the random graph $\G$ has bounded average degree.
This is arguably the most interesting case because the discrepancy between the planted and the minimum bisection gets larger as the graphs get sparser.
In fact, it is easy to see that in the case of fixed $\frac n2 \ppm=\dpm$ the difference between the planted and the minimum bisection width  is $\Theta(n)$
	as the planted bisection is not even locally optimal \whp\

Although we build upon some of the insights from~\cite{Bisect}, it seems difficult to prove our main result by tracing the 
fairly complicated algorithm from that paper.
Instead, our main tool is an elegant message passing algorithm called {\em Warning Propagation} that plays an important
role in the study of random constraint satisfaction problems via ideas from statistical physics~\cite{MM}.  
Running Warning Propagation on $\G$ naturally corresponds to a fixed point problem on the 2-simplex, and the minimum bisection width
can be cast as a function of the fixed point.

\subsection{The main result}\label{Sec_result}
To state the fixed point problem, we consider the functions
	\begin{align*}
	\psi&:\RR\to\RR,\qquad x\mapsto\begin{cases}-1&\mbox{ if }x<-1\\
			x&\mbox{ if }-1\leq x\leq 1\\
			1&\mbox{ if }x>1,
			\end{cases}
		&\tilde\psi&:\RR\to\RR,\qquad
			x\mapsto\begin{cases}-1&\mbox{ if }x\leq-1\\
			1&\mbox{ if }x>-1.
			\end{cases}
	\end{align*}
Let $\cP(\cbc{-1,0,1})$ be the set of probability measures on $\{-1,0,1\}$.
Clearly, we can identify $\cP(\cbc{-1,0,1})$ with the set of all maps $p:\{-1,0,1\}\to[0,1]$ such that $p(-1)+p(0)+p(1)=1$, i.e., the 2-simplex.
Further, let us define a map
	\begin{equation}\label{eqT}
	\cT_{\dplus ,\dminus }:\cP(\cbc{-1,0,1})\to\cP(\cbc{-1,0,1})
	\end{equation}
as follows.
Given $p\in \cP(\cbc{-1,0,1})$, let $(\eta_{p,i})_{i\geq1}$ be a family of i.i.d.\ $\{-1,0,1\}$-valued random variables with distribution $p$.
Moreover, let $\gpm=\Po(\dpm)$ 
be Poisson variables that are independent of each other and of the $\eta_{p,i}$.
Let
	\begin{equation}\label{eqZ}
	Z_{p,\dplus ,\dminus }:=\sum_{i=1}^{\gplus}\eta_{p,i}-\sum_{i=\gplus +1}^{\gplus+\gminus}\eta_{p,i}.
	\end{equation}
Then we let $\cT_{\dplus ,\dminus }(p)\in\cP(\{-1,0,1\})$ be the distribution of $\psi(Z_{p,\dplus ,\dminus })$.
Further, with $(\eta_{p,i})_{i\ge 1}$ and $\gpm$ as before, let
	\begin{align*}
	\varphi_{\dplus ,\dminus }&:\cP(\cbc{-1,0,1})\to\RR,&
	p&\mapsto\frac12\Erw\brk{
            \sum_{i=1}^{\gplus}\vecone{\cbc{\eta_{p,i}=-\tilde\psi(Z_{p,\dplus ,\dminus )}}}
            +\sum_{i=\gplus+1}^{\gplus+\gminus}\vecone{\cbc{\eta_{p,i}=\tilde\psi(Z_{p,\dplus ,\dminus )}} }}.
           \end{align*}
Moreover, let us call $p\in\cP(\cbc{-1,0,1})$ \emph{skewed} if $p(1)\geq 1-\dplus ^{-10}$.
Finally, we denote the minimum bisection width of a graph $G$ by $\bis(G)$.

\begin{theorem}\label{Thm_main}
There exists a constant $c>0$ such for any $\dpm>0$ satisfying
	$\dplus -\dminus \geq c\sqrt{\dplus \ln \dplus }$
the map $\cT_{\dplus ,\dminus }$ has a unique skewed fixed point $p^*$
and $n^{-1}\bis(\G)$ converges in probability to $\varphi_{d_+,d_-}(p^*)$.
\end{theorem}
In the following sections we will use that the assumptions of Theorem~\ref{Thm_main} allow us to assume that also $d_+$
is sufficiently large.

\subsection{Further related work}
Determining the minimum bisection width of a graph is NP-hard~\cite{GJS} and there is evidence that the problem does not even admit a PTAS~\cite{Khot}.
On the positive side, it is possible to approximate the minimum bisection width within a factor of $O(\ln n)$ for graphs on $n$ vertices in polynomial time~\cite{Harry}.

The planted bisection model has been studied in statistics under the name ``stochastic block model'' \cite{Holland}.
However, in the context of statistical inference the aim is to recover the planted partition $\SIGMA$ as best as possible given $\G$
rather than to determine the minimum bisection width.
Recently there has been a lot of progress, much of it inspired by non-rigorous work~\cite{Decelle}, on the statistical inference problem.
The current status of the problem is that matching upper and lower bounds are known for the values of $\dpm$ for which
it is possible to obtain a partition that is non-trivially correlated with $\SIGMA$ \cite{Massoulie,MNS12,MNS13}.
Furthermore, there are algorithms that recover a best possible approximation to $\SIGMA$ under certain conditions on $\dpm$~\cite{Abbe,MNS14,MNS15}.
But since our objective is different, the methods employed in the present paper are somewhat different and, indeed, rather simpler.

Finally, there has been recent progress on determining the minimum bisection width on the \Erdos-\Renyi\ random graph.
Although its precise asymptotics remain unknown in the case of bounded average degrees $d$, it was proved in~\cite{DemboMontanariSen}
that the main correction term corresponds to the ``Parisi formula'' in the Sherrington-Kirkpartrick model~\cite{Talagrand}.
Additionally, regarding the case of very sparse random graphs,
there is a sharp threshold for the minimum bisection width to be linear in $n$~\cite{LuczakMcD}.

Generally speaking, the approach that we pursue is somewhat related to the notion of ``local weak convergence'' of graph sequences as it was used in~\cite{Aldous}.
More specifically, we are going to argue that the minimum bisection width of $\G$ is governed by the ``limiting local structure'' of the graph,
which  is a two-type Galton-Watson tree.
The fixed point problem in \Thm~\ref{Thm_main} mirrors the execution of a message passing algorithm on the Galton-Watson tree.
The study of this fixed point problem, for which we use the {\em contraction method}~\cite{Ralph}, is the key technical ingredient of our proof.
We believe that this strategy provides an elegant framework for tackling many other problems in the theory of random graphs as well. In fact, in a recent paper~\cite{Cores} we combined Warning Propagation with a fixed point analysis on Galton-Watson trees to the k-core problem and in~\cite{Coloring} Warning Propagation was applied to the random graph coloring problem.

\section{Outline}

\noindent{\em From here on we keep the notation and the assumptions of \Thm~\ref{Thm_main}.
	In particular, we assume that $\dplus -\dminus \geq c\sqrt{\dplus \ln \dplus }$ for a large enough constant $c>0$
	and that $\dpm$ remain fixed as $n\to\infty$.  
	Furthermore we assume that $d_+$ is bounded from below by a large enough constant. 
	Throughout the paper all graphs will be locally finite and of countable size.
}

\bigskip\noindent
Three main insights enable the proof of \Thm~\ref{Thm_main}.
The first one, which we borrow from~\cite{Bisect}, is that \whp\ $\G$ features a fairly large set $\core$ of vertices such
that for any two optimal bisections $\tau_1,\tau_2$ of $\G$ (i.e.\ maps $\tau_1,\tau_2:V(\G)\to \{\pm 1\}$), we either have $\tau_1(v)=\tau_2(v)$ for all $v\in \core$ or $\tau_1(v)=-\tau_2(v)$
for all $v\in \core$.
In the language of random constraint satisfaction problems, the vertices in $\core$ are ``frozen''.
While there remain $\Omega(n)$ unfrozen vertices, the subgraph that they induce is subcritical, i.e., all components are of size $O(\ln n)$
and indeed most are of bounded size.

The second main ingredient is an efficient message passing algorithm called {\em Warning Propagation}, (cf.~\cite[Chapter~19]{MM}).
We will show that a bounded number of Warning Propagation iterations suffice to arrange almost all of the unfrozen vertices optimally
and thus to obtain a very good approximation to the minimum bisection \whp\ (Proposition~\ref{Prop_WP}).
This insight reduces our task to tracing Warning Propagation for a bounded number of rounds.

This last problem can be solved by studying Warning Propagation on a suitable Galton-Watson tree,
because $\G$ only contains a negligible number of short cycles \whp\ (Lemma~\ref{Lem_local}).
Thus, the analysis of Warning Propagation on the random tree is the third main ingredient of the proof.
This task will turn out to be equivalent to studying the fixed point problem from \Sec~\ref{Sec_result} (Proposition~\ref{Prop_fix}).
We proceed to outline the three main components of the proof.

\subsection{The core}\label{Sec_core}
Given a vertex $u$ of a graph $G$ let $\partial_{G}u$ denote the neighbourhood of $u$ in $G$. We sometimes omit the subscript $G$ when the graph is clear from the context.
More particularly, in the random graph $\G$, let $\partialpm u$ denote the set of all neighbours $w$ of $u$ in $\G$ with $\SIGMA(w)\SIGMA(v)=\pm1$.
Following~\cite{Bisect}, we define $\core$ as the largest subset $U\subset[n]$ such that 
	\begin{align}\label{eqCoreProperty}
	||\partialpm u|-\dpm|\leq \frac c4\sqrt{\dplus \ln \dplus }&\quad\mbox{and }\quad
		|\partial u\setminus U|\leq100\quad\mbox{for all }u\in U.
	\end{align}
Clearly, the set $\cC$, which we call the {\em core}, is uniquely defined because any union of sets $U$ that satisfy~(\ref{eqCoreProperty}) also has the property.
Let $\SIGMA_{\cC}:\cC\to\{\pm1\}$, $v\mapsto\SIGMA(v)$ be the restriction of the ``planted assignment'' to $\cC$.

Furthermore, 
for a graph $G$,
a set $U\subset V(G)$ and a map $\sigma:U\to\cbc{-1,1}$ we let
	$$\cut(G,\sigma):=\min\cbc{\sum_{\cbc{v,w}\in E(G)}\frac{1-\tau(v)\tau(w)}2\quad \bigg| \quad \tau : V(G)\to \cbc{\pm1}
		\mbox{ satisfies }\tau(v)=\sigma(v)\mbox{ for all }v\in U}.$$
In words, $\cut(G,\sigma)$ is the smallest number of edges in a cut of $G$ that separates the vertices in $U\cap\sigma^{-1}(-1)$ from
those in $U\cap\sigma^{-1}(1)$.
In particular, $\cut(G,\sigma_{\cC})$ is the smallest cut of $G$ that separates the vertices in the core $\cC$ that are frozen to $-1$ from those that are frozen to $1$.

Finally,  for any vertex $v$ we define a set $\cC_v=\cC_v(G,\sigma)$ of vertices via the following process.
	\begin{description}
	\item[C1] Let $\cC_v^{\bc 0}=\{v\}\cup\partial_{G} v$.
	\item[C2] Inductively, let $\cC_v^{(t+1)}=\cC_v^{(t)}\cup\bigcup_{u\in \cC_v^{(t)}\setminus\cC}\partial_{G} u$
		and let $\cC_v=\bigcup_{t\geq0}\cC_v^{(t)}$.
	\end{description}

\begin{lemma}[{\cite{Bisect}}]\label{Lemma_core}
We have $\bis(\G)=\cut(\G,\SIGMA_\cC)$ and $|\core|\geq n(1-\dplus ^{-100})$ \whp\
Furthermore, for any $\eps>0$ there exists $\omega>0$ such that \whp\ $\sum_{v\in\brk n}|\cC_v|\cdot\vecone\cbc{|\cC_v|\geq\omega}\leq\eps n$.
\end{lemma}

\subsection{Warning Propagation}
To calculate $\cut(\G,\SIGMA_{\cC})$ we adopt the {\em Warning Propagation} (``WP'') message passing algorithm\footnote{%
	A discussion of Warning Propagation in the context of the ``cavity method'' from statistical physics can be found in~\cite{MM}.}.
Let us first introduce WP for a generic graph $G=(V(G),E(G))$ and a map $\sigma:U\subset V(G)\to\{-1,1\}$.
At each time $t\geq0$, WP sends a ``message'' $\mu_{v\to w}(t|G,\sigma)\in\{-1,0,1\}$ from $v$ to $w$ for any edge $\{v,w\}\in E(G)$.
The messages are directed objects, i.e., $\mu_{v\to w}(t|G,\sigma)$ and $\mu_{w\to v}(t|G,\sigma)$ may differ.
They are defined inductively by
	\begin{align}\label{eqWP}
	\mu_{v\to w}(0|G,\sigma)&:=\begin{cases}\sigma(v)&\mbox{ if }v\in U,\\0&\mbox{otherwise},\end{cases}&
	\mu_{v\to w}(t+1|G,\sigma)&:=\psi\bc{\sum_{u\in\partial v\setminus w}\mu_{u\to v}(t|G,\sigma)}.
	\end{align}
Thus, the WP messages are initialised according to $\sigma$.
Subsequently, $v$ sends message $\pm1$ to $w$ if it receives more $\pm1$ than $\mp1$ messages from its neighbours $u\neq w$.
If there is a tie, $v$ sends out $0$.
Finally, for $t\ge 0$ define
	\begin{align*}
	\mu_{v}(t|G,\sigma)&:=\sum_{w\in\partial v}\mu_{w\to v}(t|G,\sigma).
	\end{align*}

\begin{proposition}\label{Prop_WP}
For any $\eps>0$ there exists $t_0=t_0(\eps,\dplus ,\dminus )$ such that for all $t\ge t_0$ \whp
	\begin{align*}
	\abs{\cut(\G,\SIGMA_\cC)-
		\frac12\sum_{v\in[n]}\sum_{w\in\partial_{\G} v}\vecone{\cbc{\mu_{w\to v}(t|\G,\SIGMA)=-\tilde\psi\br{\mu_v(t|\G,\SIGMA)}}}}\leq\eps n.
	\end{align*}
\end{proposition}

\noindent
We defer the proof of \Prop~\ref{Prop_WP} to \Sec~\ref{Sec_Prop_WP}.

\subsection{The local structure}
\Prop~\ref{Prop_WP} shows that \whp\ in order to approximate $\cut(\G,\SIGMA_\cC)$ up to a small error of $\eps n$ we merely need to run WP for
a number $t_0$ of rounds that is bounded in terms of $\eps$.
The upshot is that the WP messages $\mu_{w\to v}(t|\G,\SIGMA)$ that are required to figure out the minimum bisection width
are determined by the {\em local} structure of~$\G$.
We show that the local structure of $\G$ ``converges to'' a suitable Galton-Watson tree. For this purpose, for simplicity we always say that the number of potential neighbours of any vertex in each class is $n/2$. This ignores the fact that if $n$ is odd the classes do not have quite this size and the fact that a vertex cannot be adjacent to itself. However, ignoring these difficulties will not affect our calculations in any significant way.

Our task boils down to studying WP on that Galton-Watson tree.
Specifically, let $\T=\T_{\dplus ,\dminus }$ be the Galton-Watson tree with two types $+1,-1$ and offspring matrix
	\begin{equation}\label{eqGW}
	\begin{pmatrix}\Po(\dplus )&\Po(\dminus )\\\Po(\dminus )&\Po(\dplus )\end{pmatrix}.
	\end{equation}
Hence, a vertex of type $\pm1$ spawns $\Po(\dplus )$ vertices of type $\pm1$ and independently $\Po(\dminus )$ vertices of type $\mp1$.
Moreover, the type of the root vertex $r_{\T}$ is chosen uniformly at random.
Let $\TAU=\TAU_{\dplus ,\dminus }:V(\T)\to\{\pm1\}$ assign each vertex of $\T$ its type.

The random graph $(\G,\SIGMA)$ ``converges to'' $(\T,\TAU)$ in the following sense.
For two triples $(G,r,\sigma)$, $(G',r',\sigma')$ of graphs $G,G'$, root vertices $r\in V(G)$, $r'\in V(G')$
and maps $\sigma:V(G)\to\{\pm1\}$, $\sigma':V(G')\to\{\pm1\}$
we write $(G,\sigma)\ism(G',\sigma')$ if there is a graph isomorphism $\varphi:G\to G'$ such that 
$\varphi(r)=r'$ and $\sigma=\sigma'\circ \varphi$.
Further, we denote by $\partial^t(G,r,\sigma)$ the rooted graph obtained from $(G,r)$ by deleting all vertices at distance
greater than $t$ from $r$ together with the restriction of $\sigma$ to this subgraph.
The following lemma characterises the local structure of $(\G,\SIGMA)$.

\begin{lemma}\label{Lem_local}
Let $t>0$ be an integer and let $T$ be any tree with root $r$ and map $\tau:V(T)\to \{\pm 1\}$.
Then
	$$\frac1n\sum_{v\in[n]}\vecone\cbc{\partial^t(\G,v,\SIGMA)\ism\partial^t(T,r,\tau)}\quad \stacksign{$n\to\infty$}{\to}\quad
			\pr\brk{\partial^t(\T,r_{\T},\TAU)\ism\partial^t(T,r,\tau)}\quad\mbox{in probability}.$$
Furthermore, \whp\
	$\G$ does not contain more than $\ln n$ vertices $v$ such that
	$\partial^t(\G,v,\SIGMA)$ contains a cycle.
\end{lemma}

\noindent
\begin{proof}
 Given a tree $T$ with root $r$ and map $\tau:V(T)\to \{\pm 1\}$, let
 \begin{align*}
 X_t=X_t(T,r,\tau)=\frac1n\sum_{v\in[n]}\vecone\cbc{\partial^t(\G,v,\SIGMA)\ism\partial^t(T,r,\tau)}
 \end{align*}
and
 \begin{align*}
 p_t=p_t(T,r,\tau)=\pr\brk{\partial^t(\T,r_{\T},\TAU)\ism\partial^t(T,r,\tau)}.
 \end{align*}
The proof proceeds by induction on $t$. If $t=0$, pick a vertex $\vec v\in [n]$ uniformly at random, then
$X_0=\pr_{\vec v}\br{\SIGMA(\vec v)=\tau(r)}=\frac 12$ and $p_0=\pr_{\T}\br{\TAU(r_{\T})=\tau(r)}=\frac 12$ for any 
$\tau(r)\in\{\pm1\}.$ To proceed from $t$ to $t+1$, let $d$ denote the number of children $v_1,\ldots,v_d$ of $r$
in $T$. For each $i=1,\ldots,d$, let $T_i$ denote the tree rooted at $v_i$ in the forest obtained from $T$ by removing $r$ and 
let $\tau_i:V(T_i)\to\{\pm1\}$ denote the restriction of $\tau$ to the vertex set of $T_i$. Finally, let
$C_1,\ldots, C_{\tilde d}$ for some $\tilde d\leq d$ denote the distinct isomorphism classes 
among $\{\partial^t(T_i,v_i,\tau_i):i=1,\ldots,d\}$, and let $c_j=|\{i:\partial^t(T_i,v_i,\tau_i)\in C_j\}|$. Let $v\in[n]$
be an arbitrary vertex in $\G$. Our aim is to determine the probability of the event 
$\{\partial^{t+1}(\G,v,\SIGMA)\ism\partial^{t+1}(T,r,\tau)\}$. Therefore, we think of $\G$ as being created in three rounds.
First, partition $[n]$ in two classes. Second, randomly insert edges between vertices in $[n]\setminus\{v\}$
according to their planted sign. Finally, reveal the neighbours of $v$.
For the above event to happen, $v$ must have $d$ neighbours in $\G$.
Since $|\partial_{\pm}v|$ are independent binomially distributed random variables with parameters $\frac n2$ and $p_{\pm}$ 
and because $\frac n2 p_{\pm}= d_{\pm}$, we may approximate $|\partial_{\pm}v|$ with a poisson distribution, and $v$ has degree $d$
with probability
  \begin{align*}
    \frac{(d_++d_-)^d}{d!\exp(d_++d_-)}+\mbox{o}(1).
  \end{align*}
Conditioned on $v$ having degree $d$, by induction $v$ is adjacent to precisely $c_j$ vertices 
with neighbourhood isomorphic to $\partial^t(T_i,v_i,\tau_i)\in C_j$ with probability
  \begin{equation*}
    \binom{d}{c_1\ldots c_{\tilde d}}\prod_{j=1}^{\tilde d}p_{t}(C_j) + \mbox{o}(1).
  \end{equation*}
The number of cycles of length $\ell \le 2t+3$ in $\G$ is stochastically bounded by the number of such cycles in $\G(n,d_+/n)$ (the standard $1$-type binomial random graph). For each $\ell$, this number tends in distribution to a poisson variable with bounded mean (see e.g. Theorem~3.19 in~\cite{JLR}) and so the total number of such cycles is bounded
\whp\ Thus all the pairwise distances (in $\G-v$) between neighbours of $v$ are at least $2t+1$ \whp\ (and in particular this proves the second part of the lemma).
Therefore
  \begin{align*}
    \Erw_{\G}[X_{t+1}]=\frac{(d_++d_-)^d}{d!\exp(d_++d_-)}
    \binom{d}{c_1\ldots c_{\tilde d}}\prod_{j=1}^{\tilde d}p_{t}(C_j) + \mbox{o}(1).
  \end{align*}
By definition of $\T$, we obtain $\Erw[X_{t+1}]=p_{t+1}+\mbox{o}(1).$ To apply Chebyshev's inequality, it remains to
determine $\Erw[X_{t+1}^2]$. Let $\vec v,\vec w\in[n]$ be two randomly choosen vertices. Then \whp\ 
$\vec v$ and $\vec w$ have distance 
at least $2t+3$ in $\G$, conditioned on which $\partial^{t+1}(\G,\vec v,\SIGMA)$ and $\partial^{t+1}(\G,\vec w,\SIGMA)$ are independent. Therefore we obtain
  \begin{align*}
   &\pr_{\vec v,\vec w}\br{\partial^{t+1}(\G,\vec v,\SIGMA)\ism\partial^{t+1}(T,r,\tau)\wedge\partial^{t+1}(\G,\vec w,\SIGMA)\ism\partial^{t+1}(T,r,\tau)}\\
   &\quad=\pr_{\vec v}\br{\partial^{t+1}(\G,\vec v,\SIGMA)\ism\partial^{t+1}(T,r,\tau)}
    \pr_{\vec w}\br{\partial^{t+1}(\G,\vec w,\SIGMA)\ism\partial^{t+1}(T,r,\tau)}+\mbox{o}(1)\\
  \end{align*}
And finally
  \begin{align*}
    \Erw_{\G}[X_{t+1}^2]&=\frac 1n \Erw_{\G}[X_{t+1}] + \Erw_{\G}\brk{\pr_{\vec v}\br{\partial^{t+1}(\G,\vec v,\SIGMA)\ism\partial^{t+1}(T,r,\tau)}
    \pr_{\vec w}\br{\partial^{t+1}(\G,\vec w,\SIGMA)\ism\partial^{t+1}(T,r,\tau)}}+\mbox{o}(1)\\
    &= \Erw_{\G}[X_{t+1}]^2 +\mbox{o}(1).
  \end{align*}
The first assertion follows from Chebyshev's inequality. 
\end{proof}

\subsection{The fixed point}

Let $(T,r,\tau)$ be a rooted tree together with a map $\tau:V(T)\to\{\pm1\}$.
Then for any pair $v,w$ of adjacent vertices we have the WP messages $\mu_{v\to w}(t|T,\tau)$, $t\geq0$, as defined in (\ref{eqWP}).
Since we are going to be particularly interested in the messages directed towards the root, we introduce the following notation.
Given the root $r$, any vertex $v\neq r$ of $T$ has a unique parent vertex $w$ (the neighbour of $v$ on the unique path from $v$ to $r$).
Initially, let
  \begin{equation}\label{eq_tree_init}
    \mu_{v\toboss}(0|T,r,\tau)=\tau(v)
  \end{equation}
and define
	\begin{equation}\label{eq_init_boundary}
	\mu_{v\toboss}(t|T,r,\tau)=\mu_{v\to w}(t|T,\tau)
	\end{equation}
for $t>0$.
In addition, set $\mu_{r\toboss}(0|T,r,\tau)=\tau(r)$ and let
	\begin{equation}\label{eq_root_tree}
	\mu_{r\toboss}(t+1|T,r,\tau)=\psi\br{\sum_{v\in\partial_{\T} r}\mu_{v\toboss}(t|T,r,\tau)}\qquad(t\geq0)
	\end{equation}
be the message that $r$ would send to its parent if there was one.

For $p=(p(-1),p(0),p(1))\in\cP(\{-1,0,1\})$ we let $\bar p=(p(1),p(0),p(-1))$.
Remembering the map $$\cT=\cT_{\dplus ,\dminus }:\cP(\{-1,0,1\})\to\cP(\{-1,0,1\})$$
from \Sec~\ref{Sec_result} and writing $\cT^{t}$ for its $t$-fold iteration, we observe the following.

\begin{lemma}\label{Prop_swapped}
Let $p_t=\cT^{t}(0,0,1)$. 
\begin{enumerate}
\item Given that $\TAU(r_{\T})=+1$, the message
	$\mu_{r_{\T}\toboss}(t|\T,r_{\T},\TAU)$ has distribution $p_t$.
\item Given that $\TAU(r_{\T})=-1$, the message
	$\mu_{r_{\T}\toboss}(t|\T,r_{\T},\TAU)$ has distribution $\bar p_t$.
\end{enumerate}
\end{lemma}
\begin{proof}
The proof is by induction on $t$.
In the case $t=0$ the assertion holds because $\mu_{r_{\T}\toboss}(0|\T,r_{\T},\TAU)=\TAU(r_{\T})$.
Now, assume that the assertion holds for $t$.
To prove it for $t+1$, let $\cpm$ be the set of all children $v$ of $r_{\T}$ with $\TAU(r_{\T})\TAU(v)=\pm1$. 
By construction, $|\cpm|$ has distribution $\Po(\dpm)$.
Furthermore, let $(\T_v,v,\tau_v)$ signify the subtree pending on a child $v$ of $r_{\T}$.
Because $\T$ is a Galton-Watson tree, the random subtrees $\T_v$ are mutually independent.
Moreover, each $\T_v$ is distributed as a Galton-Watson tree with offspring matrix (\ref{eqGW}) and a root vertex of type
$\pm\TAU(r_{\T})$ for each $v\in \cpm$.
Therefore, by induction the message $\mu_{v\toboss}(t|\T_v,v,\TAU_v)$ has distribution $p_t$ if $\TAU(v)=1$ resp.\ $\bar p_t$ if $\TAU(v)=-1$.
As a consequence,
\begin{align*}
    \mu_{r_{\T}\toboss}(t+1|\T,r_{\T},\TAU)
    =\psi\bc{\sum_{v\in \cplus}\mu_{v\toboss}(t|\T_v,v,\TAU_v)
    +\sum_{v\in \cminus}\mu_{v\toboss}(t|\T_v,v,\TAU_v)}
\end{align*}
has distribution $p_{t+1}$ if $\TAU(r_{\T})=1$ and $\bar p_{t+1}$ otherwise.
\end{proof}

\Lem~\ref{Prop_swapped} shows that the operator $\cT$ mimics WP on the Galton-Watson tree $(\T,r_{\T},\TAU)$.
Hence, to understand the behaviour of WP after a large enough number of iterations we need to investigate the fixed point to which $\cT^{t}(0,0,1)$ converges
as $t\to\infty$.
In \Sec~\ref{Sec_fix} we will establish the following.

\begin{proposition}\label{Prop_fix}
The operator $\cT$ has a unique skewed fixed point $p^*$ and $\lim_{t\to\infty}\cT^{t}(0,0,1)=p^*$.
\end{proposition}

\begin{proof}[Proof of Theorem \ref{Thm_main}]
Consider the random variables
	\begin{align*}
	X_n&:= \frac 1n \bis(\G),&
	Y_n^{(t)}&:=\frac12\frac1n\sum_{v\in[n]}\sum_{w\in\partial_{\G} v}\vecone{\cbc{\mu_{w\to v}(t|\G,\SIGMA)=-\tilde{\psi}\br{\mu_v(t|\G,\SIGMA)}}}.
	\end{align*}
Then \Lem~\ref{Lemma_core} and \Prop~\ref{Prop_WP} imply that for any $\eps>0$,
	\begin{align}\label{eq_main1}
	\lim_{t\to\infty}\lim_{n\to\infty}\pr\brk{|X_n-Y_n^{(t)}|>\eps}&=0.
	\end{align}
By Definition~(\ref{eqWP}), $\mu_{w\to v}(t|\G,\SIGMA)$ and $\mu_v(t|\G,\SIGMA)$ are determined by $\partial^t_{\G}v$ and the initialisation 
$\mu_{u\to w}(0|\G,\SIGMA)$ for all $u,w\in\partial^t_{\G}v$, $\{u,w\}\in E(\G).$ Since (\ref{eq_init_boundary}) 
and (\ref{eq_root_tree}) match the recursive definition~(\ref{eqWP})
of $\mu_{w\to v}(t|\G,\SIGMA)$ and $\mu_v(t|\G,\SIGMA)$, 
Lemma~\ref{Lem_local} implies that for any fixed $t>0$ (as $n$ tends to infinity),
	\begin{align}\label{eq_main3}
	Y_n^{(t)}\quad&\stacksign{$n\to\infty$}{\to}\quad x^{(t)}:=\frac 12 \Erw\brk{\sum_{w\in\partial_{\T}r_{\T}}\vecone{\cbc{\mu_{w\toboss}(t|\T,r_{\T},\TAU)=
			 - \psi(\mu_{r_{\T}}(t|\T,r_{\T},\TAU))}}}
		\quad\mbox{in probability.}
	\end{align}
Now let $p^*$ denote the unique skewed fixed point of $\cT$ guaranteed by \Prop~\ref{Prop_fix}.
Since each child of $r_{\T}$ can be considered a root of an independent instance of $\T$ to which we can apply 
Lemma~\ref{Prop_swapped}, we obtain that given $(\TAU(w))_{w\in\partial r_{\T}}$ the sequence $(\mu_{w\toboss}(t|\T,r_{\T},\TAU))_{w\in\partial r_{\T}}$ converges to a sequence of independent random variables
$(\eta_w)_{w\in\partial r_{\T}}$
with distribution $p^*$ (if $\TAU(w)=1$) and $\bar p^*$ (if $\TAU(w)=-1$). By definition
$\mu_{r_{\T}}(t|\T,r_{\T},\TAU)$ converges to 
$\sum_{w\in\partial r_{\T},\TAU(w)=1}\eta_w+\sum_{w\in\partial r_{\T},\TAU (w)=-1}\eta_w$. 
Considering the offspring distributions of $r_{\T}$ in both cases, i.e. $\TAU(r_{\T})=\pm1$,
we obtain from $\varphi_{\dplus ,\dminus }(p)=\varphi_{\dplus ,\dminus }(\bar p)$ for all $p\in\cP (\{-1,0,1\})$ that
\begin{align}\label{eq_main4}
	\lim_{t\to\infty}x^{\bc t}&=\varphi_{\dplus ,\dminus }(p^*).
	\end{align}
Finally, combining (\ref{eq_main1})--(\ref{eq_main4}) completes the proof.
\end{proof}

\section{Proof of \Prop~\ref{Prop_WP}}\label{Sec_Prop_WP}

\begin{lemma}\label{core_fixed}
If $v\in\core$ and $w\in\partial_{\G} v$, then $\mu_{v\to w}(t|\G,\SIGMA)=\SIGMA(v)=\mu_{v\to w}(t|\G,\SIGMA_{\cC})$ for all $t\geq 0$.
\end{lemma}
\begin{proof}
We proceed by induction on $t$.
For $t=0$ the assertion is immediate from the initialisation of the messages.
To go from $t$ to $t+1$, consider $v\in\core$ and $w\in\partial_{\G}v$.
We may assume without loss of generality that $\SIGMA(v)=1$.
By the definition of the WP message,
	\begin{align}\label{eqcore_fixed1}
	\mu_{v\to w}(t+1|\G,\SIGMA)&=\psi\bc{\sum_{u\in\partial_{\G}v\setminus\cbc w}\mu_{u\to v}(t|\G,\SIGMA)}
		=\psi\bc{\splus+\sminus+S_0}
	\end{align}
where
	\begin{align*}
	\splus&:=\sum_{u\in\cC\cap\SIGMA^{-1}(+1)\cap\partial_{\G}v\setminus\cbc w}\mu_{u\to v}(t|\G,\SIGMA),\nonumber\\
	\sminus&:=\sum_{u\in\cC\cap\SIGMA^{-1}(-1)\cap\partial_{\G}v\setminus\cbc w}\mu_{u\to v}(t|\G,\SIGMA),\nonumber\\
	S_0&:=\sum_{u\in\partial_{\G}v\setminus(\cC\cup\cbc w)}\mu_{u\to v}(t|\G,\SIGMA).\nonumber
	\end{align*}
Now, (\ref{eqCoreProperty}) ensures that
	\begin{align}\label{eqcore_fixed2}
	\splus&\geq \dplus -\frac c4\sqrt{\dplus \ln \dplus },&
		\sminus&\geq -\dminus -\frac c4\sqrt{\dplus \ln \dplus },&
		|S_{0}|&\leq100\leq\frac c4\sqrt{\dplus \ln \dplus },
	\end{align}
provided that the constant $c>0$ is chosen large enough.
Combining (\ref{eqcore_fixed1}) and (\ref{eqcore_fixed2}), we see that $S_{+}+S_{-}+S_0\geq1$ and thus $\mu_{v\to w}(t+1|\G,\SIGMA)=1$.
The exact same argument works for $\mu_{v\to w}(t+1|\G,\SIGMA_{\cC})=1$.
\end{proof}

Let $\G_v$ denote the subgraph of $\G$ induced on $\cC_v$.
To prove \Prop~\ref{Prop_WP}, fix $s>0$ large enough.
Let $\cS=\cS(s)$ be the set of all vertices such that either $|\cC_v|>\sqrt s$ or $\G_v$ is cyclic.
Then \Lem~\ref{Lemma_core} (with slightly smaller $\eps$) and Lemma~\ref{Lem_local} imply that $|\cS|\leq\eps n$ \whp\
For the rest of this section, let $v\not\in\cS$ be fixed.

For $w\in\cC_v\setminus\cbc v$ we let $w_{\uparrow v}$ be the neighbour of $w$ on the path from $w$ to $v$.
We define $\G_{w\to v}$ as the component of $w$ in the graph obtained from $\G_v$ by removing the edge $\{w,w_{\uparrow v}\}$.
The vertex set of $\G_{w\to v}$ will be denoted by $\cC_{w\to v}$.
Further, $h_{w\to v}$ is the maximum distance between $w$ and any other vertex in $\G_{w\to v}$.
Additionally, $h_v$ is the maximum distance between $v$ and any other vertex in~$\G_v$.
Finally, let $\SIGMA_v:\cC_v\to\{\pm1\}$, $w\mapsto\SIGMA(w)$ and let $\SIGMA_{\cC,v}:\cC_v \cap \cC\to\{\pm1\}$, $w\mapsto\SIGMA_{\cC}(w)$.

\begin{lemma}\label{lem_fixed_m}
\begin{enumerate}
\item For any $w\in \cC_v\setminus\cbc v$ and any $t> h_{w\to v}$ we have\\
     $\mu_{w\to w_{\uparrow v}}(t|\G,\SIGMA)=\mu_{w\to w_{\uparrow v}}(h_{w\to v}+1|\G,\SIGMA)
		=\mu_{w\to w_{\uparrow v}}(t|\G,\SIGMA_{\cC}).$
\item For any $t\geq h_v$ we have 
      $\mu_{v}(t|\G,\SIGMA)=\mu_{v}(h_v+1|\G,\SIGMA)=\mu_{v}(t|\G,\SIGMA_{\cC})$.
      
\end{enumerate}
\end{lemma}
\begin{proof}
The proof of (1) proceeds by induction on $h_{w\to v}$. The construction {\bf C1--C2} of $\cC_v$ ensures that any $w\in\cC_v$ with $h_{w\to v}=0$ either belongs
to $\cC$ or has no neighbour besides $w_{\toboss v}$. Hence for the first case the assumption follows from Lemma~\ref{core_fixed}. If 
$\partial_{\G}w\setminus\{w_{\toboss v}\}=\emptyset$ we obtain that
$
 \mu_{w\to w_{\toboss v}}(t|\G,\SIGMA)=\mu_{w\to w_{\toboss v}}(t|\G,\SIGMA_{\cC})=0
$
for all $t\geq 1$ by the definition of the WP messages. Now, assume that $h_{w\to v}>0$ and let $t>h_{w\to v}$. 
Then all neighbours $u\neq w_{\toboss v}$ of $w$ in $\G_{w\to v}$ satisfy $h_{u\to v}<h_{w\to v}$.
Thus, by induction 
\begin{align*}
 \mu_{w\to w_{\toboss v}}(t|\G,\SIGMA) &=\psi\br{\sum_{u\in\partial_{\G}w\setminus\{w_{\toboss v}\}}\mu_{u\to w}(t-1|\G,\SIGMA)}\\
&=\psi\br{\sum_{u\in\partial_{\G}w\setminus\{w_{\toboss v}\}}\mu_{u\to w}(h_{u\to v}+1|\G,\SIGMA)}
=\mu_{w\to w_{\toboss v}}(h_{w\to v}+1|\G,\SIGMA).
\end{align*}
An analogous argument applies to $\mu_{w\to w_{\uparrow v}}(t|\G,\SIGMA_{\cC})$.
The proof of (2) is similar.
\end{proof}

For each vertex $w\in\cC_v$, $w\neq v$, let $\mu_{w\to v}^*=\mu_{w\to w_{\toboss v}}(s|\G,\SIGMA)$.
Further, let $\mu_w^*=\mu_w(s|\G,\SIGMA)$. 
In addition, for $z\in \{\pm 1\}$ let
	$$\SIGMA^z_{w\to v}:\cC_{w\to v}\cap(\cbc w\cup\core)\to\cbc{\pm1},\qquad
		u\mapsto\begin{cases}
			z&\mbox{ if }u=w,\\
			\SIGMA(u)&\mbox{ otherwise.}
			\end{cases}$$
In words, $\SIGMA^z_{w\to v}$ freezes $w$ to $z$ and all other $u\in\cC_{w\to v}$ that belong to the core to $\SIGMA(u)$.
Analogously, let
	$$\SIGMA^z_{v}:\cC_{v}\cap(\cbc v\cup\core)\to\cbc{\pm1},\qquad
		u\mapsto\begin{cases}
			z&\mbox{ if }u=v,\\
			\SIGMA(u)&\mbox{ otherwise.}
			\end{cases}$$

\begin{lemma}\label{Lem_dyn}
Suppose that
$u\in\cC_v\setminus\cbc v$, such that $h_{u\to v}\geq 1$.
\begin{enumerate}
\item If $z=\mu_{u\to v}^*\in\cbc{-1,1}$, then
	\begin{equation}\label{eq_cut_1}\cut(\G_{u\to v},\SIGMA^z_{u\to v})<\cut(\G_{u\to v},\SIGMA^{-z}_{u\to v}).\end{equation}
	Similarly, if $z=\psi(\mu_{v}^*)\in\cbc{-1,1}$, then
	\begin{equation}\label{eq_cut_1a}\cut(\G_{v},\SIGMA^z_{v})<\cut(\G_{v},\SIGMA^{-z}_{v}).\end{equation}
\item If $\mu_{u\to v}^*=0$, then 
	\begin{equation}\label{eq_cut_2}\cut(\G_{u\to v},\SIGMA^{+1}_{u\to v})=\cut(\G_{u\to v},\SIGMA^{-1}_{u\to v}).\end{equation}
	Similarly, if $\mu_{v}^*=0$, then
	\begin{equation}\label{eq_cut_2a}\cut(\G_{v},\SIGMA^{+1}_{v})=\cut(\G_{v},\SIGMA^{-1}_{v}).\end{equation}
\end{enumerate}
\end{lemma}
\begin{proof}
We prove (\ref{eq_cut_1}) and (\ref{eq_cut_2}) by induction on $h_{u\to v}$. If $h_{u\to v}=1$
then we have that all neighbours $w\in\partial_{\cC_{u\to v}}u$ of $u$ with $\mu^*_{u\to v}\neq 0$ are in $\cC$, i.e. fixed under $\SIGMA_{u\to v}^z$. Since 
$\cC_{u\to v}=\partial_{\G}u\setminus\{u_{\toboss v}\}\cup\{u\}$, we obtain
	\begin{align}\label{opt_1}
	\cut(\cC_{u\to v},\SIGMA^{-z}_{u\to v})-\cut(\cC_{u\to v},\SIGMA^z_{u\to v})&=
		\abs{\sum_{w\in\partial_{\G}u\setminus\{u_{\toboss v}\}}\mu_{w\to v}^*}
	\end{align}
by definition of $z$. By the induction hypothesis and because $\G_{u\to v}$ is a tree (as $v\not\in\cS$) we have that (\ref{opt_1}) holds for $h_{u\to v}>1$ as well.
A similar argument yields (\ref{eq_cut_1a}) and (\ref{eq_cut_2a}).
\end{proof}

Now, let $\cU_v$ be the set of all $w\in\cC_v$ such that $\mu_{w\to v}^*\neq 0$.
Furthermore, let
	$$\SIGMA_{\toboss v}:\cU_v\cup\{v\}\to\cbc{-1,+1},
		\qquad w\mapsto\begin{cases}
			\tilde\psi(\mu_v^*)&\mbox{ if }w=v,\\
			\mu_{w\to v}^*&\mbox{ otherwise}.
			\end{cases}$$
Thus, $\SIGMA_{\toboss v}$ sets all  $w\in\cC_v\cap\core\setminus\{v\}$ to their planted sign and all $w\in\cU_v\setminus\core$ to $\mu_{w\to v}^*$.
Moreover, $\SIGMA_{\toboss v}$ sets $v$ to $\psi(\mu_v^*)$ if $\psi(\mu_v^*)\neq 0$ and to $1$ if there is a tie.

\begin{corollary}\label{Cor_dyn}
We have $\cut(\G_v,\SIGMA_{\cC})=\cut(\G_v,\SIGMA_{\toboss v})$.

\end{corollary}
\begin{proof}
 This is immediate from Lemma \ref{Lem_dyn}. 
\end{proof}

Hence, in order to determine an optimal cut of $\G_v$ we merely need to figure out the assignment of the vertices in $\cC_v\setminus(\{v\}\cup\cU_v)$.
Suppose that $\SIGMA_{\toboss v}^*:\cC_v\to\{\pm1\}$ is an optimal extension of $\SIGMA_{\uparrow v}$ to a cut of $\G_v$, i.e.,
  $$
   \cut(\G_v,\SIGMA_{\toboss v})=\sum_{\{u,w\}\in E(\G_v)}\frac12(1-\SIGMA_{v\toboss}^*(u)\SIGMA_{v\toboss}^*(w)).
  $$

\begin{corollary}\label{Lem_first_level}
It holds that
	$\sum_{w\in\partial_{\G}v}\frac12(1-\SIGMA_{v\toboss}^*(v)\SIGMA_{v\toboss}^*(w))
		 =\sum_{w\in\partial_{\G} v}\vecone{\cbc{\mu_{w\to v}^*=-\tilde\psi\br{\mu_v}}}.$
\end{corollary}
\begin{proof}
Part (2) of Lemma \ref{Lem_dyn} implies that $\SIGMA_{v\toboss}^*(v)\SIGMA_{v\toboss}^*(w)=1$ for all $w\in\partial_{\G}v$ such that $\mu_{w\to v}^*=0$.
\end{proof}

\begin{proof}[Proof of \Prop~\ref{Prop_WP}]
Given $\eps>0$ choose $\delta=\delta(\eps,\dplus ,\dminus )$ sufficiently small and $s=s(\eps,\delta,\dplus ,\dminus )>0$ sufficiently large.
In particular, pick $s$ large enough so that 
	\begin{equation}\label{eqProp_WP1}
	\pr\br{|\cS|\geq\delta n}<\eps.
	\end{equation}
Provided that $\delta$ is suitable small, the Chernoff bound implies that  for  large $n$
\begin{align}\label{eqProp_WP2}
  \pr\br{\frac 12 \left.\sum_{v\in\cS}|\partial_{\G}v|\geq\eps n \right||\cS|<\delta n}
 	&<\eps.
\end{align}
Now, suppose that $\SIGMA^*_{\cC}$ is an optimal extension of $\SIGMA_{\cC}$ to a cut of $\G$ and let $v\not\in\cS$. Then 
using the definition of $\cC_v$, Corollary~\ref{Cor_dyn} implies
that
\begin{align*}
 \sum_{w\in\partial_{\G} v}(1-\SIGMA_{\cC}^*(v)\SIGMA_{\cC}^*(w))
 =\sum_{w\in\partial_{\G} v}(1-\SIGMA_{v\toboss}^*(v)\SIGMA_{v\toboss}^*(w)).
\end{align*}
Therefore, we obtain
\begin{align*}
 \pr\br{ \abs{\cut(\G,\SIGMA_{\cC})-\frac12\sum_{v\not\in\cS}
 \sum_{w\in\partial_{\G} v}(1-\SIGMA_{v\toboss}^*(v)\SIGMA_{v\toboss}^*(w))}
  \geq \eps n}
\leq\pr\br{\frac 12 \sum_{v\in\cS}|\partial_{\G}v|\geq\eps n}\leq 2\eps.
\end{align*}
The assertion follows from Lemma~\ref{lem_fixed_m} for $t\geq s$. 
\end{proof}

\section{Proof of \Prop~\ref{Prop_fix}}\label{Sec_fix}

\noindent
We continue to denote the set of probability measures on $\cX\subset\RR^k$ by $\cP(\cX)$.
For a $\cX$-valued random variable $X$ we denote by $\cL(X)\in\cP(\cX)$ the distribution of $X$.
Furthermore, if $p,q\in\cP(\cX)$, then $\cP_{p,q}(\cX)$ denotes the set of all probability measures $\mu$ on $\cX\times\cX$
such that the marginal distribution of the first (resp.\ second) component coincides with $p$ (resp.\ $q$).
The space $\cP(\{-1,0,1\})$
is complete with respect to (any and in particular) the $L_1$-Wasserstein metric, defined by
	\begin{equation*}
	\ell_1(p,q)=\inf\cbc{\Erw|X-Y|:X,Y\mbox{ are random variables with }
		\cL(X,Y)\in\cP_{p,q}(\{-1,0,1\})}.
	\end{equation*}
In words, the infimum of $\Erw|X-Y|$ is over all couplings $(X,Y)$ of the distributions $p,q$.
Such a coupling $(X,Y)$ is {\em optimal} if
$\ell_1(p,q)=\Erw|X-Y|$.
Finally,
let $\cP^*(\cbc{-1,0,1})$ be the set of all skewed probability measures on $\cbc{-1,0,1}$.
Being a closed subset of  $\cP(\cbc{-1,0,1})$, 
 $\cP^*(\cbc{-1,0,1})$ is complete with respect to $\ell_1(\nix,\nix)$.

As in the definition~\eqref{eqT}-\eqref{eqZ} of the operator $\cT=\cT_{\dplus,\dminus}$ for $p\in\cP(\{-1,0,1\})$ we let $(\eta_{p,i})_{i\ge1}$ be a family of independent
random variables with distribution $p$.
Further,  let $\gpm=\Po(\dpm)$ be independent of each other and of the the $(\eta_{p,i})_{i\ge1}$.
We introduce the shorthands
	$$Z_p=Z_{p,\dplus ,\dminus },\qquad
		\zpplus=\sum_{i=1}^{\gplus}\eta_{p,i}
		,\qquad \zpminus=\sum_{i=\gplus +1}^{\gplus+\gminus}\eta_{p,i}\quad\mbox{so that}\quad
			Z_p=\zpplus-\zpminus.$$
Also set $\lambda=c\sqrt{\dplus \ln \dplus }$ and recall that $c>0$ is a constant that we assume to be sufficiently large.

\begin{lemma}\label{Prop_fixedpoint1}
The operator $\cT$ maps $\cP^*(\cbc{-1,0,1})$ into itself.
\end{lemma}
\begin{proof}
Suppose that $p\in\cP(\{-1,0,1\})$ is skewed.
Then
\begin{align}\label{eqIntoItself1}
 \pr\left(Z_p< 1\right)
 &\leq\pr\br{Z_{p,+}\leq \dplus -\frac{\lambda-1}2}
   +\pr\br{Z_{p,-}\geq \dminus +\frac{\lambda-1}2}.
\end{align}
Since $|\eta_{p,i}|\leq 1$ for all $i$, we can bound the second summand from above by 
invoking the Chernoff bound to obtain
\begin{align}\label{eqIntoItself2}
 \pr\br{\gminus \geq \dminus +\frac c2\sqrt{\dplus \ln \dplus }-\frac 12}&
 	< \frac 13 \dplus ^{-10},
\end{align}
provided $c$ is large enough.
To bound the other summand from above we use that $(\eta_{p,i})_{i\geq 1}$ is a sequence of independent {\em skewed} random variables, whence by the Chernoff bound
\begin{align}\nonumber
 \pr\br{\zpplus\leq \dplus -\frac{\lambda-1}2}
&\leq
	\pr\br{\abs{\gplus-\dplus }>\lambda/8} + \pr\br{\zpminus\leq \dplus -\frac{\lambda-1}2\bigg|\gplus\geq \dplus -\lambda/8}
\\
&\leq 	\frac 13 \dplus ^{-10}+\pr\brk{\Bin(d_+-\lambda/8,1-\dplus ^{-10})\leq \dplus -\lambda/7}
<\frac 23 d_+^{-10},\label{eqIntoItself3}
\end{align}
provided that $c$ is sufficiently big.
Combining (\ref{eqIntoItself1})--(\ref{eqIntoItself3}) completes the proof.
\end{proof}

\begin{lemma}\label{Prop_fixedpoint2}
The operator $\cT$ is $\ell_1$-contracting on $\cP^*(\cbc{-1,0,1})$.
\end{lemma}
\begin{proof}
Let $p, q\in\mathcal P^*(\cbc{-1,0,1})$.
We aim to show that  $\ell_1(\cT(p),\cT(q))\leq\frac12\ell_1(p,q)$.
To this end, we let $(\eta_{p,i},\eta_{q,i})_{i\geq 1}$ be a family of random variables with distribution $p$ resp.\ $q$
such that $(\eta_{p,i})_{i\geq 1}$ are independent and $(\eta_{q,i})_{i\geq 1}$ are independent
but such that the pair $(\eta_{p,i},\eta_{q,i})$ is an optimal coupling for every $i$.
Then by the definition of $\ell_1(\nix,\nix)$,
	\begin{align}\label{FP_bound_1}
	\ell_1(\cT(p),\cT(q))&\leq\Erw\abs{\psi(Z_p)-\psi(Z_q)}.
	\end{align}
To estimate the r.h.s., let $\tilde\eta_{p,i}=\vecone\{\eta_{p,i}=1\}$, $\tilde\eta_{q,i}=\vecone\{\eta_{q,i}=1\}$.
Further, let $\mathfrak F_i$ be the $\sigma$-algebra generated by $\tilde\eta_{p,i},\tilde\eta_{q,i}$ and let $\mathfrak F$ be the $\sigma$-algebra
generated by $\gplus,\gminus$ and the random variables $(\tilde\eta_{p,i},\tilde\eta_{q,i})_{i\geq1}$.
Additionally, let $\gamma=\gplus+\gminus$ and consider the three events
\begin{align*}
\mathfrak A_1&=\left\{\sum_{i=1}^{\gamma}\tilde\eta_{p,i}\tilde\eta_{q,i}\geq\gamma-10\right\},&
\mathfrak A_2&=\left\{\gamma\geq 2\dplus \right\},&
\mathfrak A_3&=\left\{\gplus-\gminus\leq 20\right\}.
\end{align*}
We are going to bound
$|\psi(Z_p)-\psi(Z_q)|$ on $\mathfrak A_1\setminus(\mathfrak A_2\cup\mathfrak A_3)$,
	$\overline{\mathfrak A_1\cup \mathfrak A_2\cup{\mathfrak A_3}}$,
	$\mathfrak A_2$ and $\mathfrak A_3\setminus\mathfrak A_2$ separately.
The bound on the first event is immediate:
if $\mathfrak A_1\setminus(\mathfrak A_2\cup\mathfrak A_3)$ occurs, then $\psi(Z_p)=\psi(Z_q)=1$ with certainty.
Hence,
\begin{align}\label{eqKathrin0} 
\Erw\brk{|\psi(Z_p)-\psi(Z_q)|\cdot \vecone_{\mathfrak A_1\setminus(\mathfrak A_2\cup\mathfrak A_3)}}=0.
\end{align}

Let us turn to the second event $\overline{\mathfrak A_1\cup \mathfrak A_2\cup{\mathfrak A_3}}$.
Because the pairs $(\eta_{p,i},\eta_{q,i})_{i\geq1}$ are mutually independent, we find
	\begin{align}\label{eqKathrin1a}
	\mathbb E\left[\left.\left|\eta_{p,i}-\eta_{q,i}\right|\right|\mathfrak F\right]
		  =\mathbb E\left[\left.\left|\eta_{p,i}-\eta_{q,i}\right|\right|\mathfrak F_i\right]\qquad\mbox{for all }i\ge 1.
		\end{align}
Clearly, if $\tilde\eta_{p,i}\tilde\eta_{q,i}=1$, then $\eta_{p,i}-\eta_{q,i}=0$.
Consequently,
	\begin{align}\label{eqKathrin1}
	\Erw\left[\left.\left|\eta_{p,i}-\eta_{q,i}\right|\right|\mathfrak F_i\right]
		&\leq\frac{\Erw|\eta_{p,i}-\eta_{q,i}|}{\pr[\tilde\eta_{p,i}\tilde\eta_{q,i}=0]}
			=\frac{\Erw|\eta_{p,1}-\eta_{q,1}|}{\pr[\tilde\eta_{p,1}\tilde\eta_{q,1}=0]}.
	\end{align}
Since the events $\mathfrak A_1,\mathfrak A_2,\mathfrak A_3$ are $\mathfrak F$-measurable and because $\bar{\mathfrak A}_2$  ensures that $\gamma<2\dplus $,
	(\ref{eqKathrin1a}) and (\ref{eqKathrin1}) yield
\begin{align}\label{eqKathrin3}
\mathbb E[|\psi(Z_p)-\psi(Z_q)|\,|\mathfrak F]\vecone_{\overline{\mathfrak A_1\cup \mathfrak A_2\cup{\mathfrak A_3}}}
    \leq\frac{ 2\dplus  \Erw|\eta_{p,1}-\eta_{q,1}|}{\pr[\tilde\eta_{p,1}\tilde\eta_{q,1}=0]}
 \cdot   \vecone_{\overline{\mathfrak A_1\cup \mathfrak A_2\cup{\mathfrak A_3}}}.
    \end{align}
Further, because the pairs $(\eta_{p,i},\eta_{q,i})_{i\ge 1}$ are independent and because $p,q$ are skewed,
\begin{align}\label{eqKathrin4}
\pr\br{\overline{\mathfrak A_1\cup \mathfrak A_2\cup{\mathfrak A_3}}}
&\leq\pr\br{\gamma\leq2\dplus ,\sum_{i=1}^{\gamma}\tilde\eta_{p,i}\tilde\eta_{q,i}\leq\gamma-10}
\leq \left(2\dplus \pr\left(\tilde{\eta}_{p,1}\tilde{\eta}_{q,1}=0\right)\right)^{10}.
\end{align}
Combining (\ref{eqKathrin3}) and (\ref{eqKathrin4}), we obtain
\begin{align}\label{eqKathrin5}
\Erw\brk{\Erw\brk{|\psi(Z_p)-\psi(Z_q)||\mathfrak F}\vecone_{\overline{\mathfrak A_1\cup \mathfrak A_2\cup{\mathfrak A_3}}}}
\leq(2d_+)^{11}\pr\left(\tilde{\eta}_{p,1}\tilde{\eta}_{q,1}=0\right)^9\Erw|\eta_{p,1}-\eta_{q,1}|.
\end{align}
Since $p,q$ are skewed, we furthermore obtain $\pr\br{\tilde{\eta}_{p,1}\tilde{\eta}_{q,1}=0}\leq 2d_+^{-10}$. Therefore
\begin{align*}
\Erw\brk{|\psi(Z_p)-\psi(Z_q)|\vecone_{\overline{\mathfrak A_1\cup \mathfrak A_2\cup{\mathfrak A_3}}}}
&=
\Erw\brk{\Erw\brk{|\psi(Z_p)-\psi(Z_q)||\mathfrak F}\vecone_{\overline{\mathfrak A_1\cup \mathfrak A_2\cup{\mathfrak A_3}}}}
\leq 2^{20}\dplus ^{-79}\Erw|\eta_{p,1}-\eta_{q,1}|.
\end{align*}
With respect to $\mathfrak A_2$, the triangle inequality yields
	\begin{align}\label{eqKathrin10}
	\Erw[|\psi(Z_p)-\psi(Z_q)|\vecone_{\mathfrak A_2}]&\leq2\Erw|\eta_{p,1}-\eta_{q,1}|\cdot\Erw[\gamma\vecone_{\mathfrak A_2}].
	\end{align}
Further, since $\gamma=\Po(\dplus +\dminus )$, the Chernoff bound entails that $\Erw[\gamma\vecone_{\mathfrak A_2}]\leq \dplus ^{-1}$ if the constant $c$
is chosen large enough.
Combining this estimate with (\ref{eqKathrin10}), we get
	\begin{align}\label{eqKathrin11}
	\Erw[|\psi(Z_p)-\psi(Z_q)|\vecone_{\mathfrak A_2}]&\leq2\dplus ^{-1}\Erw|\eta_{p,1}-\eta_{q,1}|.
	\end{align}

Finally, on $\mathfrak A_3\setminus\mathfrak A_2$ we have
	\begin{align}\label{eqKathrin20}
	\Erw[|\psi(Z_p)-\psi(Z_q)|\vecone_{\mathfrak A_3\setminus\mathfrak A_2}]&\leq4\dplus \Erw|\eta_{p,1}-\eta_{q,1}|
		\pr\brk{\gplus-\gminus\leq20}.
	\end{align}
Since $\gpm=\Po(\dpm)$ and $\dplus -\dminus \geq\lambda$, the Chernoff bound yields $\pr\brk{\gplus-\gminus\leq20}\leq \dplus ^{-2}$,
if $c$ is large enough.
Hence, (\ref{eqKathrin20}) implies
	\begin{align}\label{eqKathrin21}
	\Erw[|\psi(Z_p)-\psi(Z_q)|\vecone_{\mathfrak A_3\setminus\mathfrak A_2}]&\leq4\dplus ^{-1}\Erw|\eta_{p,1}-\eta_{q,1}|.
	\end{align}
Finally, the assertion follows from (\ref{FP_bound_1}), (\ref{eqKathrin0}),
(\ref{eqKathrin5}),	(\ref{eqKathrin11})	and (\ref{eqKathrin21}).
\end{proof}

\begin{proof}[Proof of \Prop~\ref{Prop_fix}]
The assertion follows from \Lem s~\ref{Prop_fixedpoint1} and~\ref{Prop_fixedpoint2} and the Banach fixed point theorem.
\end{proof}


\begin{thebibliography}{29}


\bibitem{Abbe}
E.\ Abbe, A.\ Bandeira, G.\ Hall. Exact recovery in the stochastic block model. arxiv 1405.3267 (2014).

\bibitem{Aldous}
D.\ Aldous , J.\ Steele: The objective method: probabilistic combinatorial optimization and local weak convergence (2003).
In: H.\ Kesten (ed.): Probability on discrete structures, Springer 2004.

\bibitem{ARV}
S.\ Arora, S.\ Rao, U.\ Vazirani:
Expander flows, geometric embeddings and graph partitioning.
J.\ ACM {\bf  56} (2009) 5.

\bibitem{Coloring}
V.\ Bapst, A.\ Coja-Oghlan, S.\ Hetterich, F.\ Rassmann, D.\ Vilenchik: The condensation phase transition in random graph coloring.
Submitted. arXiv: 1404.5513 (2014).

\bibitem{BBScott} 
B.\ Bollob\'as, A.\ Scott: Max cut for random graphs with a planted partition.
Combinatorics, Probability and Computing {\bf13} (2004) 451--474.

\bibitem{Boppana}
R.\ Boppana: Eigenvalues and graph bisection: an average-case analysis. Proc.~28th FOCS (1987) 280--285.

\bibitem{Bui}
T.\ Bui, S.\ Chaudhuri, T.\ Leighton, M.\ Sipser:
Graph bisection algorithms with good average case behavior.
Combinatorica {\bf7} (1987) 171--191.

\bibitem{Carson}
T.\ Carson, R.\ Impagliazzo: Hill-climbing finds random planted bisections.
Proc.~12th SODA (2001) 903--909.

\bibitem{Bisect}
A.\ Coja-Oghlan: A spectral heuristic for bisecting random graphs.
Random Structures \& Algorithms	{\bf29} (2006) 351--398.

\bibitem{Cores}
A.\ Coja-Oghlan, O.\ Cooley, M.\ Kang, K.\ Skubch: How does the core sit inside the mantle? Submitted. arXiv: 1503.09030 (2015).

\bibitem{CondonKarp}
A.\ Condon, R.\ Karp: Algorithms for graph partitioning on the planted partition model.
Random Structures \&\ Algorithms {\bf18} (2001) 116--140.

\bibitem{Decelle}
A.\ Decelle, F.\ Krzakala, C.\ Moore,  L.\ Zdeborov\'a:
Asymptotic analysis of the stochastic block model for modular networks and its algorithmic applications.
Physics Review E {\bf 84} (2011) 066106.

\bibitem{DemboMontanariSen}
A.\ Dembo, A.\ Montanari, S.\ Sen:
Extremal cuts of sparse random graphs.
arXiv:1503.03923 (2015).

\bibitem{Dimitriou}
T.\ Dimitriou, R.\ Impagliazzo:
Go with the winners for graph bisection.
Proc.~9th SODA (1998) 510--520.


\bibitem{DyerFrieze}
M.\ Dyer, A.\ Frieze: The solution of some random NP-hard problems in polynomial expected time.
J.~Algorithms {\bf10} (1989) 451--489.

\bibitem{Feige}
U.\ Feige, J.\ Kilian:
Heuristics for semirandom graph problems.
J.\ Comput.\ System Sci.\ {\bf 63} (2001) 639--671.

\bibitem{FKBisect}
U.\ Feige, R.\ Krauthgamer:
A polylogarithmic approximation of the minimum bisection.
SIAM J.~Comput.~{\bf31} (2002) 1090-1118

\bibitem{Ofek}
U.\ Feige, E.\ Ofek: Spectral techniques applied to sparse random graphs.
Random Structures \&\ Algorithms {\bf 27} (2005) 251--275.

\bibitem{GJS}
M.\ Garey, D.\ Johnson, L.\ Stockmeyer:
Some simplified NP-complete graph problems.
Theoret.\ Comp.\ Sci.\ {\bf1} (1976) 237--267.

\bibitem{GW}
M.\ Goemans, D.\ Williamson:
Improved approximation algorithms for maximum cut and satisfiability problems using semidefinite programming. J.\ ACM {\bf 42} (1995) 1115--1145.

\bibitem{Holland}
P.\ Holland, K.\ Laskey, S.\ Leinhardt: Stochastic blockmodels: first steps.
Social Networks {\bf 5} (1983) 109--137.

\bibitem{JLR}
S.\ Janson, T.\ {\L}uczak, A.\ Ruci\'nski: Random Graphs. Wiley (2000).

\bibitem{JerrumSorkin}
M.\ Jerrum, G.\ Sorkin: The Metropolis algorithm for graph bisection.
Discr.\ Appl.\ Math.\ {\bf82} (1998) 155--175.

\bibitem{Juels}
A.\ Juels:
Topics in black-box combinatorial function optimization.
Ph.D.\ thesis. UC Berkeley (1996).

\bibitem{Karp}
R.\ Karp: Reducibility among combinatorial problems.
In: R.\ Miller, J.\ Thatcher (eds): Complexity of computer computations (1972) 85--103.

\bibitem{Karpinski}
M.\ Karpinski: Approximability of the minimum bisection problem: an algorithmic challenge.
Proc.~27th MFCS (2002) 59--67.

\bibitem{Khot}
S.\ Khot:
Ruling out PTAS for graph min-bisection, densest subgraph and bipartite clique.
Proc.~45th FOCS (2004) 136--145.

\bibitem{Kucera}
L.\ Ku\v{c}era: Expected complexity of graph partitioning problems.
Discrete Applied Mathematics {\bf 57} (1995) 193--212.

\bibitem{LuczakMcD}
M.\ Luczak, C.\ McDiarmid:
Bisecting sparse random graphs.
Random Structures \&\ Algorithms {\bf18} (2001) 31--38.

\bibitem{McSherry}
F.\ McSherry: Spectral partitioning of random graphs.
Proc.~42nd FOCS (2001) 529--537.

\bibitem{MM}
M.~M\'ezard, A.~Montanari:
Information, physics and computation.
Oxford University Press~2009.

\bibitem{Makarychev}
K.\ Makarychev, Y.\ Makarychev, A.\ Vijayaraghavan:
Approximation algorithms for semi-random partitioning problems.
Proc.\ 44th STOC (2012).
In Proceedings of the forty-fourth annual ACM symposium on Theory of computing. ACM, 2012. 367--384



\bibitem{Massoulie} 
L.\ Massouli\'e:
Community detection thresholds and the weak Ramanujan property.
Proc.\ 46th STOC (2014) 694--703.

\bibitem{MNS14}
E.\ Mossel, J.\ Neeman, A.\ Sly:
Belief propagation, robust reconstruction and optimal recovery of block models.
Proc.\ 27th COLT (2014) 356--370.

\bibitem{MNS12}
E.\ Mossel, J.\ Neeman, A.\ Sly: Stochastic block models and reconstruction.
arXiv 1202.1499 (2012).

\bibitem{MNS13}
E.\ Mossel, J.\ Neeman, A.\ Sly: A proof of the block model threshold conjecture.
arXiv 1311.4115 (2013).

\bibitem{MNS15}
E.\ Mossel, J.\ Neeman, A.\ Sly:
Consistency thresholds for the planted bisection model.
arXiv 1407.1591 (2014).

\bibitem{Ralph}
R.\ Neininger, L.\ R\"uschendorf:
A general limit theorem for recursive algorithms and combinatorial structures.
The Annals of Applied Probability {\bf 14} (2004) 378--418. 


\bibitem{Harry}
H.\ R\"acke:  Optimal hierarchical decompositions for congestion minimization in networks.
Proc.\ 40th STOC (2008) 255--264.

\bibitem{Talagrand}
M.\ Talagrand: The Parisi formula.
Annals of Mathematics {\bf 163} (2006) 221--263.


\end{thebibliography}
\end{document}